\documentclass[10pt]{article}
\usepackage{graphicx, verbatim, url, amssymb, amsmath, amsfonts, amsthm, latexsym, enumerate}
\newtheorem{definition}{Definition}

\newtheorem{example}{Example}
\newtheorem{corollary}{Corollary}
\newtheorem{lemma}{Lemma}
\newtheorem{theorem}{Theorem}

\begin{document}

\title{Arbitrarily regularizable graphs}

\author{Enrico Bozzo \\
Department of Mathematics, Computer Science and Physics \\ 
University of Udine \\
\url{enrico.bozzo@uniud.it} \and 
Massimo Franceschet\\
Department of Mathematics, Computer Science and Physics \\ 
University of Udine \\
\url{massimo.franceschet@uniud.it}}

\maketitle

\begin{abstract}
A graph is regularizable if it is possible to assign weights to its edges so that all nodes have the same degree. Weights can be positive, nonnegative or arbitrary as soon as the regularization degree is not null. Positive and nonnegative regularizable graphs have been thoroughly investigated in the literature. In this work, we propose and study arbitrarily regularizable graphs. In particular, we investigate necessary and sufficient regularization conditions on the topology of the graph and of the corresponding adjacency matrix. Moreover, we study the computational complexity of the regularization problem and characterize it as a linear programming model.
\end{abstract}

\section{Introduction}\label{Introduction}

In this introduction, we first provide a colloquial account to the problem of regularizability. Next, we provide an application scenario to the regularization problem in the context of social and economic networks. Finally, we outline our contribution.

\subsection{An informal account}

Consider a network of nodes and edges. The nodes represent the network actors and the edges are the connections between actors. Another ingredient is the force of the connection between two connected actors. We can represent this force with a number, called the weight associated with each edge: the greater the weight, the stronger the relation between the actors linked by the edge. The \textit{strength} of a node as the sum of the weights of the edges connecting the node to some neighbor. The strength of a node gives a rough estimate of how much important is the node.

The \textit{regularization problem} we approach in this work is the following: given a network, is there an assignment of weight to edges such that all nodes have the same non-zero strength value? Hence, we seek for a way to associate a level of force with the connections among actors so that the resulting network is found to be egalitarian, that is, all players have the same importance.\footnote{In a regularized weighted network, all nodes have the same strength, the same eigenvector centrality and the same power, as defined in \cite{BF16}.} As an example, Figure \ref{fig:gasreg} depicts a subset of the European natural gas pipeline network. Nodes are European countries (country codes according to ISO 3166-1) and there is an undirected edge between two nations if there exists a natural gas pipeline that crosses the borders of the two countries. Edge weights, represented by widths of lines, are such such each country has the same strength.

Some networks, like star graphs, are not regularizable at all. Others, namely regular graphs, are immediately regularizable. In between, we have a hierarchy of graph classes: those networks that are regularizable with arbitrarily, possibly negative, weights, those that are regularizable with nonnegative, possibly null, weights, and those that are regularizable with strictly positive weights. In this paper, we  investigate the following two meaningful questions:

\begin{enumerate}
\item what is the topology of networks that are regularizable?
\item if a network is regularizable, how do we find the regularization weights for edges and how complex is to find them?
\end{enumerate}

\begin{figure}[t]
\begin{center}
\includegraphics[scale=0.4, angle=0]{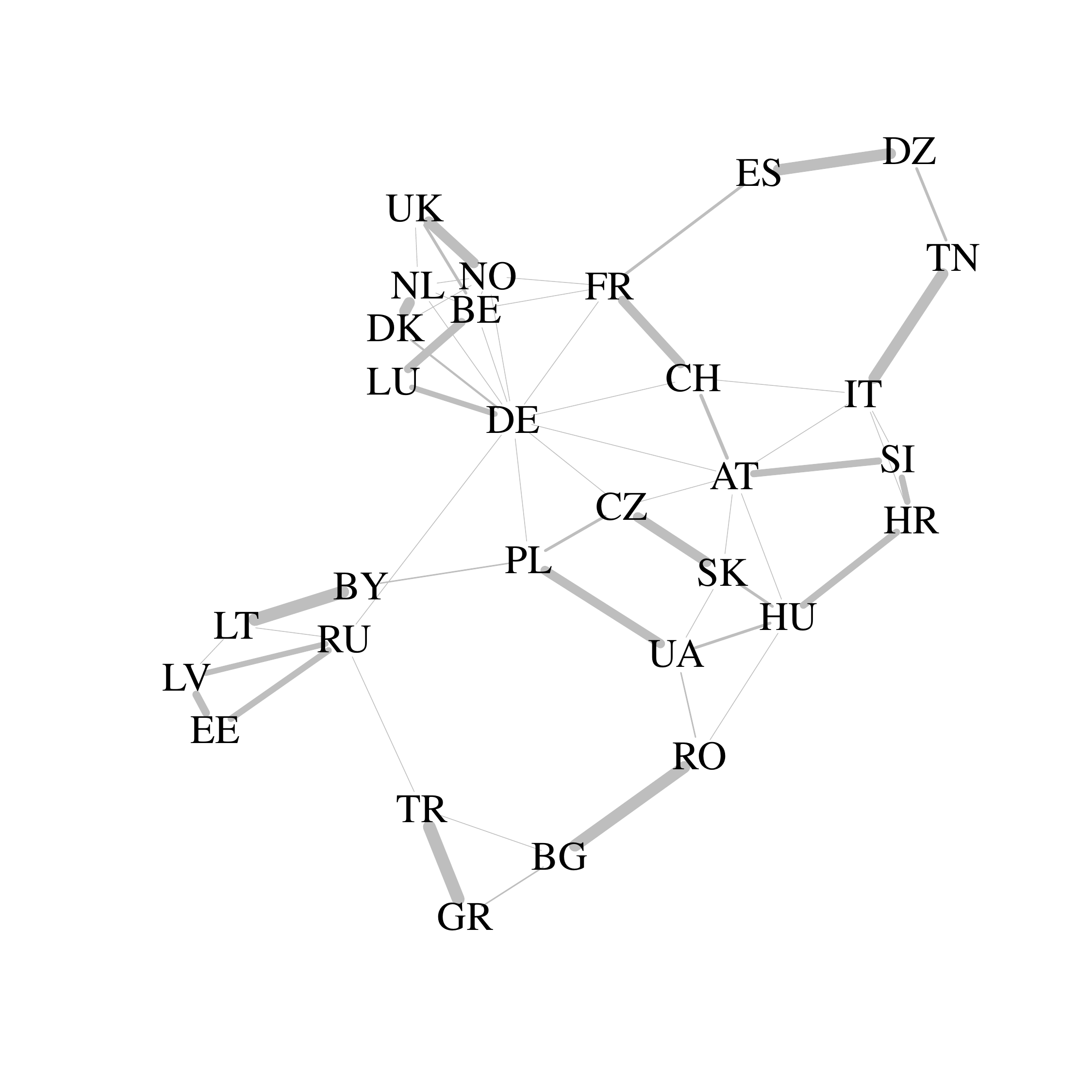}
\end{center}
\caption{The European natural gas pipeline network. Edge width, proportional to edge weight, is such that each country has the same strength.}
\label{fig:gasreg}
\end{figure}

\begin{figure}[t]
\begin{center}
\includegraphics[scale=0.5, angle=0]{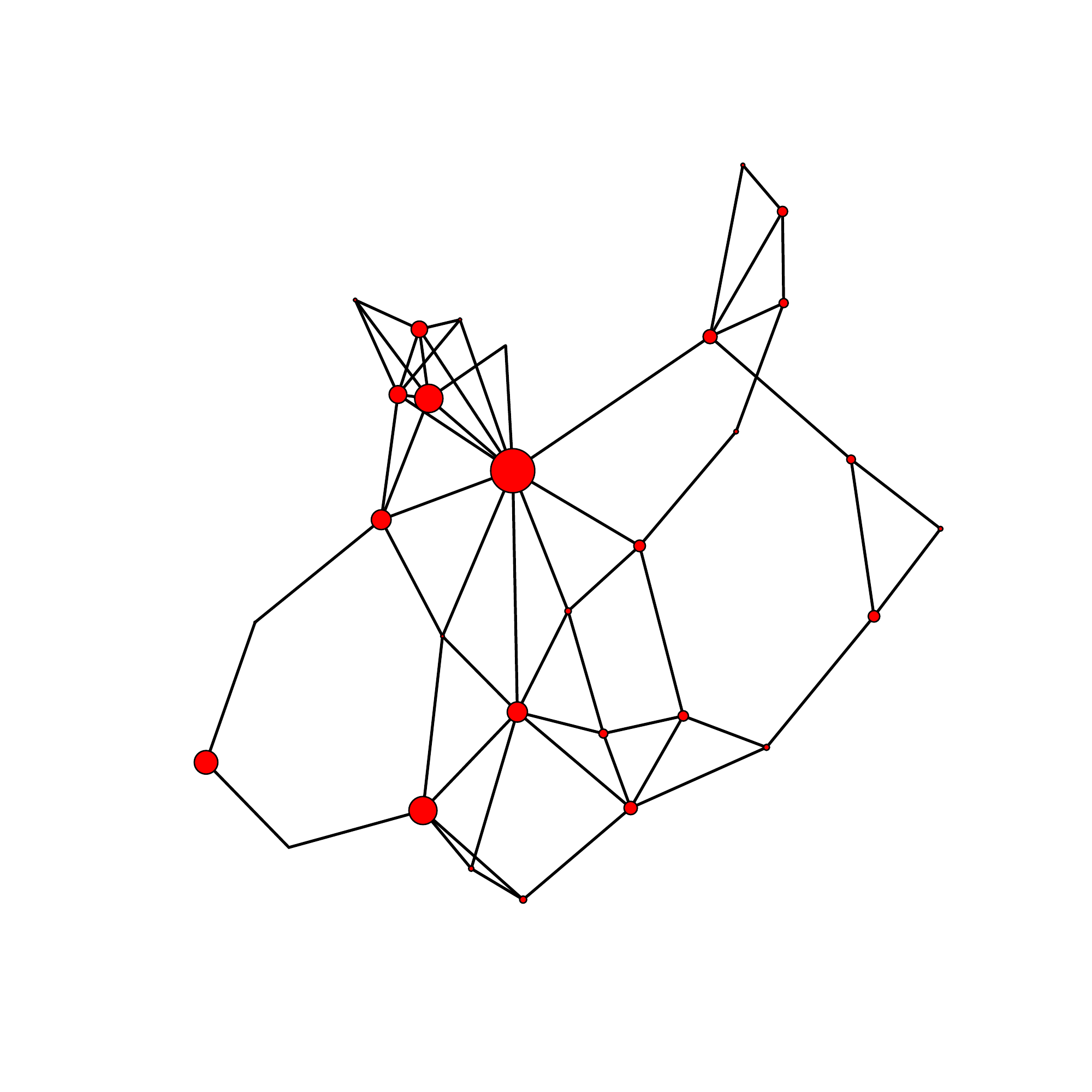}
\end{center}
\caption{The European natural gas pipeline network. Node size is proportional to its power.}
\label{fig:gasPower}
\end{figure}

\subsection{Application scenario}

In the following, inspired by \cite{EK10}, we define an application scenario for the general regularization problem in the context of social and economic networks. Consider a set of economic or social actors, linked with ties in a network. Each actor has the same (economic or social) capital to spend, say a unit of value. There are two constraints about how  to spend the capital: (i) actors can spend an amount of capital only on incident edges of the network; (ii) actors must spend the whole capital (not less, not more). The network has a \textit{power equilibrium} if each actor can reach an agreement with its neighbors on the amount of capital to spend on each edge. A network is regularizable if it has a power equilibrium.

For instance, consider the simple graph A -- B -- C. Both A and C have only one option to spend their unit of capital, namely B. Hence, they will never spend less than 1 on the edge linking with B. On the other hand, B can accept only one offer either from A or from C, because it cannot spend more than 1. Hence, no equilibrium is possible on this network, and thus it is not regularizable. If we add a link from C to A, closing the path into a cycle, the three actors can now spend 0.5 on each incident edge; the equilibrium is found and the network is regularizable. Notice that this is the only way to divide the capital that reach an equilibrium. 
Moreover, the graph A -- B -- C -- D  is regularizable but only if we allow
that the capital spend on the central edge is zero. This simple example shows that relaxing the constraints on the weights give us higher flexibility with respect to the graphs that we can regularize. More sophisticated examples of regularizability are given in Section \ref{inclusion}.

The described notion of power equilibrium is in fact intimately related with the notion of \textit{power} in networks. The study of power has a long history in economics (in its acceptation of bargaining power) and sociology (in its interpretation of social power) \cite{EK10,J08}. In his seminal work on power-dependance relations, dated 1962, Richard Emerson claims that power is a property of the social relation, not an attribute of the person: ``X has power'' is vacant, unless we specify ``over whom'' \cite{E62}. This type of relational power is endogenous with respect to the network structures, meaning that it is a function of the position of the node in the network. In particular, \cite{BF16} recently proposed a notion of power that claims that an actor is powerful if it is connected with powerless actors. It is precisely implemented in with the recursive equation $p = A \, p^{\div}$, where $p$ is the sought power vector, $A$ is a matrix encoding the network and $p^{\div}$ is the vector whose entries are the reciprocal of those of $p$. A non-trivial result proves that the equation defining power has a solution on a network if and only if the network has a power equilibrium in which the division of the capital on edges is strictly positive \cite{SK67,BF16}. Figure \ref{fig:gasPower} depicts the regularizable network of gas pipelines in Europe where node size is proportional to power as computed with the above equation.

\subsection{Our contribution}

The regularization problem sketched above is customarily defined on undirected graphs and constrained to nonnegative and positive weights \cite{B78,B81}. Nonnegative and positive regularizable graphs have been already studied in the literature, in particular in connection with the \textit{balancing problem}  of a nonnegative matrix $A$, namely the question of finding diagonal matrices $D_1$ and $D_2$ so that all the rows and columns of $P = D_1 A D_2$ sum to one \cite{KR13}. Some motivations for achieving this balance, and hence for studying nonnegative and positive regularizable graphs, include interpreting economic data \cite{B70}, understanding traffic circulation \cite{LS81}, assigning seats fairly after elections \cite{B08}, matching protein samples \cite{DMGMW09} and centrality measures in networks \cite{K08,BF16}.

As accounted in \cite{TCAL16}, a recent survey with an extensive bibliography, signed graphs, that are graphs with both positive and negative links, are now attracting increasing attention and disclosing promising research directions.
Signed undirected graphs have been used already in the 1940s as a representation and analysis tool in social psychology, leading to the definition of structural balance theory. A counterpart of balance theory for directed graphs is  status theory that in its most basic form suggests that a positive link from A to B or a negative link form B to A indicates that B has an higher status than A. In modern social media links can be directed and, if positive,  indicate friendship, trust, like or the opposite whereas negative. For example, Slashdot is a technology news website that publishes stories written by editors or submitted by users and allows users to comment on them. The site has the Zoo feature, which lets users tag other users as friends and foes. In constrast to most popular social networking services, Slashdot Zoo is one of the few sites that also allows users to rate other users negatively \cite{KLB09}. In this scenario, research problems have evolved from theory development to mining tasks, very often requiring dedicated methods with respect to unsigned networks.

In this paper, we generalize the problem of regularization to arbitrary regularizable graphs, meaning that the regularization weights can be arbitrary, possibly negative, numbers. We find necessary and sufficient regularization conditions on the graph and on the corresponding adjacency matrix. To this end, we define a graph-theoretic notion of alternating path and show that it corresponds to a notion of chain on adjacency matrices \cite{HM75}. We study the computational complexity of the problem of deciding whether a graph is regularizable and model the problem of finding the regularization weights as a linear programming feasibility problem.

The rest of paper is organized as follows. In Section \ref{hierarchy} we present the regularization hierarchy and review some known results on regularizability with nonnegative and positive weights. In Section \ref{NRG} we investigate arbitrary regularizable graphs. In Section \ref{inclusion} we show that the inclusion relation among the graph classes is strict. Section \ref{complexity} is devoted to computational issues. In Section \ref{related} we present the related literature. Conclusions are drawn in Section~\ref{conclusion}.

\section{The regularization hierarchy} \label{hierarchy}

In this section we define the regularization hierarchy of graph classes and review nonnegative and positive regularizable graphs.

Let $A$ be a square $n \times n$ matrix. We denote with $G_A$ the graph whose adjacency matrix is $A$, that is, $G_A$ has $n$ nodes numbered from $1$ to $n$ and it has an edge $(i,j)$ from node $i$ to node $j$ with weight $a_{i,j}$ if and only if $a_{i,j} \neq 0$. Any square  matrix $A$ corresponds to a weighted graph $G_A$ and any weighted graph $G$, once its nodes have been ordered, corresponds to a matrix $A_G$. A \textit{permutation} matrix $P$ is a square matrix such that each row and each column of $P$ contains exactly one entry equal to $1$ and all other entries are equal to $0$. The two graphs $G_A$ and $G_{P^TAP}$ are said to be isomorphic since they differ only in the way their nodes are ordered.

If $A$ and $B$ are square matrices of the same size, then we write $A \subseteq B$ if $a_{i,j} \neq 0$ implies $b_{i,j} \neq 0$, that is, the set of non-zero entries of $A$ is a subset of the set of non-zero entries of $B$. If $A \subseteq B$, then $G_A$ is a subgraph of $G_B$. We write $A \equiv B$ if both $A \subseteq B$ and $B \subseteq A$. Hence, $A \equiv B$ means that the two matrices have the same zero/nonzero pattern, and the graphs $G_A$ and $G_B$ have the same topological structure (they may differ only for the weighing of edges).
Given $r > 0$, a matrix $W$ with nonnegative entries is $r$-\emph{bistochastic} if $W e = W^T e = r e$, where we denote with $e$ the vector of all 1's. A 1-bistochastic matrix is simply called bistochastic.

Let $G$ be a graph with $n$ nodes and $m$ edges. We enumerate the edges of $G$ from $1$ to $m$. Let $U$ be the $n \times m$ out-edges incidence matrix such that $u_{i,l} = 1$ if $l$ corresponds to edge $(i,j)$ for some $j$, that is edge $l$ exits node $i$, and $u_{i,l} = 0$ otherwise.  Similarly, let $V$ be the $n \times m$ in-edges incidence matrix such that $v_{i,l} = 1$ if $l$ corresponds to edge $(j,i)$ for some $j$, that is edge $l$ enters node $i$, and $v_{i,u} = 0$ otherwise. Consider the following $2n \times m$ incidence matrix:

$$
B =
\begin{bmatrix}
    \, U \, \\
    \, V \, \\
\end{bmatrix}
$$

Let $w = (w_1, \ldots, w_m)$ be the vector of edge weight variables and $r$ be a variable for the regularization degree. The regularization linear system is as follows:

\begin{eqnarray}
B w = r e \label{eq:system}
\end{eqnarray}

If $G$ in an undirected graph (that is, its adjacency matrix is symmetric), then there is no difference between in-edges and out-edges. In this case, the incidence matrix $B$ of $G$ is an $n \times m$ matrix such that $b_{i,l} = 1$ if $i$ belongs to edge $l$ and $b_{i,l} = 0$ otherwise.

Notice that system (\ref{eq:system}) has always the trivial solution $(w,r) = 0$.
The set of non-trivial solutions of system (\ref{eq:system}) induces the following \textit{regularization hierarchy} of graphs:

\begin{itemize}
\item \textit{Arbitrarily regularizable graphs}: those
graphs for which there exists at least one solution $w \neq 0$ and  $r > 0$ of system (\ref{eq:system}). Notice that $w$ can contain negative components but the regularization degree must be positive.

\item \textit{Nonnegatively regularizable graphs}: those for which there exists at least one solution of system (\ref{eq:system}) such that $w$ has nonnegative entries and  $r > 0$.

\item \textit{Positively regularizable graphs}: those for which there exists at least one solution of system (\ref{eq:system}) such that $w$ has positive entries and $r > 0$.

\item \textit{Regular graphs}: those graphs for which $w = e$ and $r > 0$ is a solution of system (\ref{eq:system}).

\end{itemize}

Clearly, a regular graph is a positively regularizable graph, a positively regularizable graph is a nonnegatively regularizable graph, and a nonnegatively regularizable graph is an arbitrarily regularizable graph.
In Section \ref{inclusion} we show that this inclusion is strict, meaning that each class is properly contained in the previous one, for both undirected and directed graphs.

\subsection{Positively regularizable graphs} \label{PRG}

Informally, a graph is positively regularizable
if it becomes regular by weighting its edges with positive values. More precisely, if $G$ is a graph and $A_G$ its adjacency matrix, then graph $G$ (or its adjacency matrix $A_G$) is \emph{positively regularizable} if there exists $r > 0$ and an $r$-bistochastic  matrix $W$ such that $W \equiv A_G$.
A matrix $A$ has \emph{total support} if $A \neq 0$ and for every pair $i,j$ such that $a_{i,j} \neq 0$ there is a permutation matrix $P$ with $p_{i,j} = 1$ such that $P \subseteq A$. Notice that a permutation matrix $P$ corresponds to a graph $G_P$ whose strongly connected components are cycles of length greater than or equal to 1. We call such a graph a (directed) \textit{spanning cycle forest}. Hence, a matrix $A$ has total support if each edge of $G_A$ is contained in a spanning cycle forest of $G_A$.

The following result can be found in \cite{PM65}. For the sake of completeness, we give the proof in our notation.

\begin{theorem} \label{th:regularizability}
Let $A$ be a square matrix. Then $A$ is positively regularizable if and only if $A$ has total support.
\end{theorem}

\begin{proof}
If $A$ is positively regularizable there exists $r > 0$ and an $r$-bistochastic $W$ such that $W\equiv A$. Clearly $\overline{W} = (1/r) W$ is bistochastic and has the same pattern of $A$. By Birkhoff theorem,  see for example \cite{HJ13}, we obtain $\overline{W} = \sum_i \alpha_i P_i$, where every $\alpha_i > 0$, $\sum_i \alpha_i = 1$ and every $P_i$ is a different permutation matrix. Hence, for every  $i,j$ such that $\overline{w}_{i,j} > 0$ there exists some permutation matrix $P_k$ such that $[P_k]_{i,j} = 1$ and $P_k \subseteq \overline{W}$. Since $\overline{W}\equiv A$, we conclude that $A$ has total support.

On the other hand, suppose that matrix $A$ has total support. Let $E$ be the set of non-zero entries of $A$. Then, for every non-zero entry $u \in E$ there is a permutation matrix $P_u$ with $[P_u]_{i,j} = 1$ and $P_u \subseteq A$. Let $W = \sum_{u \in E} P_u$. Notice that $W$ is nonnegative and has the same pattern than $A$. Moreover, for every $P_u$ it holds $P_u e = P_{u}^{T} e = e$, that is $P_u$ is bistochastic. Thus $W e = W^T e = m e$, where $m = |E|$, that is, $W$ is $m$-bistochastic. We conclude that $A$ is positively regularizable.
\end{proof}

From Theorem \ref{th:regularizability} and definition of total support, it follows that a graph is positively regularizable if and only if each edge is included in a spanning cycle forest. Moreover, from the proof of Theorem \ref{th:regularizability} it follows that if a graph is positively regularizable then there is a solution of the regularization system (\ref{eq:system}) with integer weights.

We now switch to the undirected case, which corresponds to symmetric adjacency matrices. Let $Q = P + P^T$, with $P$ a permutation matrix. Each element $q_{i,j}$ is either $0$ (if both $p_{i,j} = 0$ and $p_{j,i} = 0$), 1 (if either $p_{i,j} = 1$ or $p_{j,i} = 1$ but not both), or 2 (if both $p_{i,j} = 1$ and $p_{j,i} = 1$).  Notice that $Q$ is symmetric and $2$-bistochastic. Moreover, $Q$ corresponds to an undirected graph $G_Q$ whose connected components are single edges or cycles (including loops, that are cycles of length 1). We call these graphs (undirected) \textit{spanning cycle forests}. For symmetric matrices, we have the following:

\begin{corollary} \label{th:regularizability2}
Let $A$ be a symmetric square matrix. Then $A$ has total support if and only if for every pair $i,j$ such that $a_{i,j} > 0$ there is a matrix $Q=P+P^T$, with $P$ a permutation matrix, such that $q_{i,j} > 0$ and $Q \subseteq A$.

\end{corollary}

\begin{proof}
If $A$ has total support then for every pair $i,j$ such that $a_{i,j} > 0$ there is a  permutation $P$ such that $p_{i,j} = 1$ and $P \subseteq A$. Let $Q = P + P^T$. Hence $q_{i,j} > 0$ and since $A$ is symmetric $Q \subseteq A$. On the other hand, if for every pair $i,j$ such that $a_{i,j} > 0$ there is a matrix $Q = P + P^T$ such that $q_{i,j} > 0$ and $Q \subseteq A$, then $p_{i,j} > 0$ and $P \subseteq A$ (and $p_{j,i} > 0$ and $P^T \subseteq A$).

\end{proof}

Hence an undirected graph is positively regularizable if each edge is included in an undirected spanning cycle forest.

\subsection{Nonnegatively regularizable graphs} \label{NNRG}

Informally, a nonnegatively regularizable graph
is a graph that can be made regular by weighting its edges with nonnegative values. More precisely, if $G$ is a graph and $A_G$ its adjacency matrix, then graph $G$ is \emph{nonnegatively regularizable} if there exists $r > 0$ and an $r$-bistochastic matrix $W$ such that $W \subseteq A_G$.

A matrix $A$ has \emph{support} if there is a permutation matrix $P$ such that $P \subseteq A$. The following result is well-known, see for instance \cite{LP86}. For the sake of completeness, we give the proof in our notation.

\begin{theorem} \label{th:quasi-regularizability}
Let $A$ be a square matrix. Then $A$ is nonnegatively regularizable if and only if $A$ has support.
\end{theorem}

\begin{proof}
Suppose $A$ is nonnegatively regularizable. Then there is $r > 0$ and an $r$-bistochastic $W$ with $W \subseteq A$.  Since $r > 0$, we have that $\overline{W} = (1/r) W$ is  bistochastic and $\overline{W} \subseteq A$. Hence, by Birkhoff theorem, $\overline{W} = \sum_i \alpha_i P_i$, where every $\alpha_i > 0$, $\sum_i \alpha_i = 1$ and every $P_i$ is a different permutation matrix. Let $P_k$ be any permutation matrix in the sum that defines $\overline{W}$. Then $P_k \subseteq W \subseteq A$ and hence $P_k \subseteq A$. We conclude that $A$ has support.

On the other hand, suppose $A$ has support. Then there is a permutation matrix $P$ with $P \subseteq A$. Since $P$ is  bistochastic, we have that $A$ is nonnegatively regularizable.
\end{proof}

From Theorem \ref{th:quasi-regularizability} and definition of support it follows that a graph is nonnegatively regularizable if and only if it contains a spanning cycle forest. Furthermore, from the proof of Theorem \ref{th:quasi-regularizability} it follows that if a graph is nonnegatively regularizable then there exists a  solution of the regularization system with binary weights (0 and 1).

We now consider undirected graphs, that is, symmetric matrices. We have the following:
\begin{corollary} \label{th:regularizability3}
Let $A$ be a symmetric square matrix. Then $A$ has support if and only if there is a matrix $Q=P+P^T$, with $P$ a permutation matrix, with $Q \subseteq A$.
\end{corollary}

\begin{proof}
If $A$ has support then there is a permutation $P$ with $P \subseteq A$. Let $Q = P + P^T$. Since $A$ is symmetric we have $Q \subseteq A$. On the other hand, if there is a matrix $Q = P + P^T$, with $P$ permutation and $Q \subseteq A$, then $P \subseteq A$ (and $P^T \subseteq A$).
\end{proof}

Hence an undirected graph is nonnegatively regularizable if and only if it contains an undirected spanning cycle forest. Moreover, if a graph is nonnegatively regularizable than there exists a regularization solution with weights 0, 1 and 2.

\section{Arbitrarily regularizable graphs} \label{NRG}
Informally, an arbitrarily regularizable graph is a graph that can be made regular by weighting its edges with arbitrarily values. More precisely, if $G$ is a graph and $A_G$ its adjacency matrix, then graph $G$ is \emph{negatively regularizable} if there exists $r > 0$ and a matrix $W$ (whose entries are not restricted to be nonnegative) such that  $We=W^Te=re$ and $W \subseteq A_G$.
Our goal here is to topologically characterize the class of arbitrarily regularizable graphs.

\subsection{The undirected case} \label{NRG:undirected}

We first address the case of undirected graphs (that is, symmetric adjacency matrices). The next result, see \cite{BA14}, will be useful.

\begin{lemma}\label{lem:rank}
Let $G$ be a connected undirected graph with $n$ nodes and let $B$ be the incidence matrix of $G$.
Then the rank of $B$ is $n-1$ if $G$ is bipartite and $n$ otherwise.
\end{lemma}

First of all, notice that an undirected graph is arbitrarily regularizable (resp., nonnegatively, positively) if and only if all its connected components are so. It follows that we can focus on undirected graphs that are connected. Let $V$ be the set of the $n$ nodes of an undirected connected graph $G$. If $G$ is bipartite then $V$ can be partitioned into two subsets $U$ and $W$ such that each edge connects a node in $U$ with a node in $W$. If $|U|=|W|$ the bipartite graph is said to be \textit{balanced}, otherwise it is called \textit{unbalanced}. Let us introduce a vector $s$, that we call the \textit{separating vector}, where the entries corresponding to the nodes of $U$ are equal to $1$ and the entries corresponding to the nodes of $V$ are equal to $-1$. Clearly we have that $s^T B=0$, where $B$ is the incidence matrix of the graph. We have the following result:

\begin{theorem} \label{th:negreg2}
Let $G$ be an undirected connected graph. Then:

\begin{enumerate}
\item if $G$ is not bipartite, then $G$ is arbitrarily regularizable;
\item if $G$ is bipartite and balanced, then $G$ is arbitrarily regularizable;
\item if $G$ is bipartite and unbalanced, then $G$ is not arbitrarily regularizable.

\end{enumerate}

\end{theorem}

\begin{proof}

\noindent

We prove item (1). By virtue of Lemma \ref{lem:rank} the incidence matrix $B$ of $G$ has  rank equal to the number of its rows.
By permuting the columns of $B$ without loss of generality we can assume that:
$$B=\begin{bmatrix} M   & N
\end{bmatrix}$$
where $M$ is $n\times n$ and nonsingular and $N$ in $n\times (m-n)$. Let $x=r M^{-1}e$ be a vector of length $n$ and let
$$ y = \begin{bmatrix} x \\ 0 \end{bmatrix}$$ be a vector of length $m$ obtained by concatenating $x$ with a vector of 0s. Notice that $x \neq 0$, hence $y \neq 0$. Then
$$By = re $$
Hence the linear system $Bw = re$ has a nontrivial solution $y \neq 0$ for every $r > 0$ so that $G$ is arbitrarily regularizable.  If $r=|\det(M)|>0$ then the vector $y$ has integer entries, since the entries of $M$ are integers.

We prove item (2). If $G$ is bipartite and balanced, then $|V|=n$ is even and $|U|=|W|=n/2$.
By Lemma \ref{lem:rank} the rank of $B$ is $n-1$ and by permuting the rows and columns of $B$, without loss of generality, we can assume that:

$$B=\begin{bmatrix} M   & N \\
                   a^T  & b^T
\end{bmatrix}$$

where $M$ is $(n-1)\times (n-1)$ and nonsingular, $N$ is $(n-1) \times (m-n+1)$, $a^T$ is $1\times (n-1)$, and $b^T$ is $1 \times (m-n+1)$. Let $r > 0$ and $x=r M^{-1} e$ be a vector of length $n-1$, and $$ y = \begin{bmatrix} x \\ 0 \end{bmatrix}$$ be a vector of length $m$ obtained by concatenating $x$ with a vector of 0s. Notice that $x \neq 0$, hence $y \neq 0$. Then

$$By = \begin{bmatrix} re \\ a^Tx \end{bmatrix}$$

If $s$ is the separating vector, then $s^T B = 0$ so that

$$ s^T B y =  s^T \begin{bmatrix} re \\ a^Tx \end{bmatrix} = 0$$

Since half entries of $s$ are equal to $1$ and the remaining half are equal to $-1$, it must be
$a^T x = r$. Hence $B y = re$ so that $y \neq 0$ and $r>0$ is a solution of system (\ref{eq:system}). Hence $G$ is arbitrarily regularizable.
Again, if $r=|\det(M)|>0$ the vector $y$ has integer entries.

We prove item (3). If $G$ is bipartite and unbalanced, then $|U| \neq |W|$. Suppose $G$ is arbitrarily regularizable. Then $Bw = re$ where $w \neq 0$ and $r>0$. Let $s$ be the separating vector. Since $s^T B = 0$ we have that $s^T B w = r s^T e = 0$, that is $r(|U|-|W|)=0$. Since $G$ is unbalanced we have $|U|\neq |W|$ and hence it must be $r=0$. Hence $G$ is not arbitrarily regularizable.
\end{proof}

From the proof of Theorem \ref{th:negreg2} it follows that if a graph is arbitrarily regularizable then there exists regularization solution with integer weights. We recall that the same holds for nonnegatively and positively regularizable graphs. A \textit{tree} is an undirected connected acyclic graph. Since acyclic, a tree is bipartite. As a corollary, a tree is arbitrarily regularizable if and only if it is balanced as a bipartite graph. The next theorem points out that acyclic and cyclic graphs behave differently when they are not arbitrarily regularizable.

\begin{theorem} \label{th:negreg1}
Let $G$ be an undirected connected graph that is not arbitrarily regularizable and let $B$ be its incidence matrix. Then
\begin{enumerate}
\item if $G$ is acyclic then the system $B w = re$ has only the trivial  solution $w=0$ and $r=0$;
\item if $G$ is cyclic then the system $B w = re$ has infinite many solutions such that $w \neq 0$ and $r = 0$.
\end{enumerate}
\end{theorem}

\begin{proof}

Let $n$ be the number of nodes and $m$ be the number of edges of $G$. Consider the homogeneous linear system $\hat{B} \hat{w} = 0$, where
$\hat{B}=\begin{bmatrix}B & -e\end{bmatrix}$ is  $n \times (m+1)$  and $\hat{w} = \begin{bmatrix} w\\ r \end{bmatrix}$.
Notice that $w=0$ implies $r=0$. Since $G$ is not arbitrarily regularizable then either there is only one trivial solution $\hat{w} =0$ or the system has infinite many solutions $\hat{w}$ different from the null one with $w \neq 0$ and $r = 0$.
Since $G$ is connected, then $m \geq n-1$. We have that:
\begin{enumerate}
\item If $G$ is acyclic then $m=n-1$ and $\textrm{rank}(B) = n-1$ by virtue of Lemma \ref{lem:rank}. Hence the columns of $B$ are linearly independent so that the system $\hat{B} \hat{w} = 0$ cannot have solutions with $w \neq 0$ and $r = 0$.

\item When $G$ has cycles, we have indeed that $m \geq n$. But $\textrm{rank}(\hat{B}) \leq n$, since $\hat{B}$ has $n$ rows. Hence $\textrm{rank}(\hat{B}) \leq n \leq m < m+1$ so that the system $\hat{B} \hat{w} = 0$ has infinite many solutions.
\end{enumerate}

\end{proof}

Hence graphs that are not arbitrarily regularizable can be partitioned in two classes: unbalanced trees, for which the only solution of system $B w = re$ is the trivial null one, and cyclic bipartite unbalanced graphs, for which there are infinite many solutions with $w \neq 0$ and $r = 0$. For instance, consider the chair graph with undirected edges $(x_1,x_2)$, $(x_2,x_3)$, $(x_3,x_4)$, $(x_4,x_1)$, $(x_5,x_1)$. It is cyclic, bipartite and unbalanced. If $(x_1,x_2)$ and $(x_3,x_4)$ are labelled with $\alpha > 0$, $(x_2,x_3)$ and $(x_4,x_1)$ are labelled with $-\alpha$, and $(x_5,x_1)$ is labelled with $0$, then all nodes have degree 0.

\subsection{The directed case} \label{NRG:directed}

We now address to the case of directed graphs.
Consider the following mapping from directed graphs to undirected graphs. If $G$ is a directed graph, let $G^*$ be its undirected counterpart such that each node $x$ of $G$ corresponds to two nodes $x_1$ (with color white) and $x_2$ (with color black) of $G^*$, and each directed edge $(x,y)$ in $G$ corresponds to the undirected edge $(x_1, y_2)$. Notice that $G^*$ is a bipartite graph with $2n$ nodes ($n$ white nodes and $n$ black nodes) and $m$ edges that tie white and black nodes together. Moreover, the degree of the white node $x_1$ (resp., black node $x_2$) of $G^*$ is the out-degree (resp., the in-degree) of the node $x$ in $G$.

Despite $G$ is weakly or strongly connected, $G^*$ can have many connected components. However, we have the following:

\begin{theorem} \label{th:negreg3}
Let $G$ be a directed graph. Then $G$ is arbitrarily regularizable (resp., nonnegatively regularizable, positively regularizable) if and only if $G^*$ is arbitrarily regularizable (resp., nonnegatively regularizable, positively regularizable).
\end{theorem}

\begin{proof}
The crucial observation is the following: if we order in $G^*$ the white nodes before the black nodes, then the incidence matrix $B_G$ of the directed graph $G$ (as defined in this section) is precisely the incidence matrix $B_{G^*}$ of $G^*$. It follows that $(w,r)$ is a solution of system (\ref{eq:system}) for $G$ if and only if $(w,r)$ is a solution of system (\ref{eq:system}) for $G^*$. Hence the thesis.
\end{proof}

Notice that the connected componentes of $G^*$ are bipartite graphs. Using Theorem \ref{th:negreg2}, we have the next result.

\begin{theorem} \label{th:negreg4}
Let $G$ be a directed  graph. Then $G$ is arbitrarily regularizable if and only if all connected components of $G^*$ are balanced (they have the same number of white and black nodes).
\end{theorem}

In the following, we provide two alternative versions of Theorem \ref{th:negreg4}, namely Theorems \ref{th:negreg5} and \ref{th:negregchar}, which serve to better characterize arbitrarily regularizable graphs. In the first rewriting of the theorem we will use the following graph-theoretic notion of alternating path.

\begin{definition} \label{def:path}
Let $G$ be a directed graph. A \textnormal{directed path} of length $k \geq 0$ is a sequence of directed edges of the form:

$$
(x_1, x_2), (x_2, x_3), (x_3, x_4), \ldots, (x_k, x_{k+1})
$$

An \textnormal{alternating path of type 1} of length $k \geq 0$ is a sequence of directed edges of the form:

$$
(x_1, x_2), (x_3, x_2), (x_3, x_4), (x_5, x_4), \ldots, (x_k, x_{k+1})
$$

If $x_1 = x_k$ we have an \textnormal{alternating cycle}.
An \textnormal{alternating path of type 2} of length $k \geq 0$ is a sequence of directed edges of the form:

$$
(x_2, x_1), (x_2, x_3), (x_4, x_3), (x_4, x_5), \ldots, (x_k, x_{k+1})
$$

If $x_1 = x_{k+1}$ we have an \textnormal{alternating cycle}.

\end{definition}

Observe that if we reverse the edges in even (resp., odd) positions of an alternating path of type 1 (resp., type 2) we get a directed path. Moreover, in simple graphs, an alternating cycle is either a self-loop or an alternating path of even length greater than or equal to 4. It is interesting to notice that if the edges of an alternating path are labelled with alternating signs, then the path induces, according to status theory \cite{TCAL16}, an ordering on the nodes of the path.

Let $G = (V,E)$ be a directed graph. We define an \textit{alternating path relation} $\leftrightarrow$ on the set $E$ of edges of $G$ such that for $e_1, e_2 \in E$, we have $e_1 \leftrightarrow e_2$ if there is an alternating path that starts with $e_1$ and ends with $e_2$. Notice that $\leftrightarrow$ is reflexive, symmetric and transitive, hence it is an equivalence relation. Thus $\leftrightarrow$ induces a partition of the set of edges $E = \bigcup_i E_i$ where $E_i$ are nonempty pairwise disjoint sets of edges. It is easy to realize that for each $i$ the edges of $E_i$ corresponds to the edges of some connected components of the undirected counterpart $G^*$ of $G$.  Each $E_i$ induces a subgraph $G_i = (V_i, E_i)$ of $G$. We say that a node in $G$ is a \textit{white} node if it has positive outdegree, it is a \textit{black} node if it has positive indegree, it is a \textit{source} node if it has null indegree, and it is a \textit{sink} node if it has null outdegree. Notice that a node can be both white and black, or neither white nor black; also, it can be both source and sink, or neither source nor sink.

We are now ready to prove the following alternative characterization of arbitrarily regularizable graphs.

\begin{theorem} \label{th:negreg5}
Let $G = (V,E)$ be a directed  graph and $E = \bigcup_i E_i$ be the partition of edges induced by the alternating path binary relation $\leftrightarrow$. Then $G$ is arbitrarily regularizable if and only if

\begin{enumerate}
\item $G$ contains neither source nor sink nodes;
\item all subgraphs $G_i = (V_i, E_i)$ induced by edge sets $E_i$ have the same number of white and black nodes.
\end{enumerate}

\end{theorem}

\begin{proof}
Suppose $G$ is arbitrarily regularizable. If $G$ contains a source or a sink node, then the regularization degree of $G$ must be $0$, hence $G$ cannot be arbitrarily regularizable. Hence we assume that $G$ contains neither source nor sink nodes. By Theorem \ref{th:negreg4} we have that all connected components of the undirected counterpart $G^*$ of $G$ are balanced (have the same number of white and black nodes). Since for each $i$ the edges of $E_i$ corresponds to the edges of some connected components $C_i$ of $G^*$, and a white node (resp., black node) in $G_i = (V_i, E_i)$ corresponds to a white node (resp., black node) in $C_i$, we have that all subgraphs $G_i = (V_i, E_i)$ induced by edge sets $E_i$ have the same number of white and black nodes.

On the other hand, if $G$ contains neither source nor sink nodes, then each connected component of $G^*$ contains at least one edge. Hence, each connected component $C_i$ in $G^*$ corresponds to some edge set $E_i$ of $G$. Since all subgraphs $G_i = (V_i, E_i)$ induced by edge sets $E_i$ have the same number of white and black nodes, and a white node (resp., black node) in $C_i$ corresponds to a white node (resp., black node) in $G_i = (V_i, E_i)$, we have that all connected components of $G^*$ are balanced, hence by Theorem \ref{th:negreg4} we have that $G$ is arbitrarily regularizable.

\end{proof}

The third, and last, characterization of arbitrarily regularizable graphs is based on the notion of matrix chainability.  First of all, we observe that, given a matrix $A$, if $P$ and $Q$ are permutation matrices then the two graphs $G_A$ and $G_{PAQ}$  are not necessarily isomorphic. However,  $G_A^*$ and $G_{PAQ}^*$
are always isomorphic. Actually, if we set $$B=\begin{bmatrix} 0 & A \\ A^T & 0\end{bmatrix},$$ then $G_A^*$ is equal to $G_B$. Moreover
$$\begin{bmatrix} 0 & PAQ \\ (PAQ)^T & 0\end{bmatrix}=
       \begin{bmatrix} P & 0 \\ 0 & Q^T\end{bmatrix}\begin{bmatrix} 0 & A \\ A^T & 0\end{bmatrix}\begin{bmatrix} P & 0 \\ 0 & Q^T\end{bmatrix}^T,$$
so that $G_B$ and $G_{PAQ}^*$ are isomorphic.

Now, we recall the definition of chainable matrix \cite{HM75,SK69}.

\begin{definition}
An $m\times n$ matrix $A$ is chainable if
\begin{enumerate}
\item $A$ has no rows or columns composed entirely of zeros, and
\item for each pair of non-zero entries $a_{p,q}$ and $a_{r,s}$ of $A$ there is a sequence of non-zero
entries $a_{i_k,j_k}$ for $k=1,\ldots,t$ such that $(p,q)=(i_1,j_1)$, $(r,s)=(i_t,j_t)$  and  for $k=1,\ldots,t-1$ either $i_k=i_{k+1}$ or $j_k=j_{k+1}$.
\end{enumerate}
\end{definition}
\noindent
As noted in \cite{HM75,SK69},  the property of being chainable can be described by saying that one may move from a nonzero entry of $A$ to another by a sequence of rook moves on the nonzero entries of the matrix.
Notice that if $A$ is the adjacency matrix of a graph $G_A=(V,E)$ then $A$ is chainable if and only if for all $e_1, e_2\in E$ it holds that $e_1\leftrightarrow e_2$. Actually an alternating path between the edges of $G_A$ corresponds to a sequence of rook moves between the nonzero entries of $A$.
It is interesting to observe that if $A$ is chainable then $A^T$ is chainable. In addition, if $A$ is chainable and $P$ and $Q$ are permutation matrices then $PAQ$ is chainable, since if two entries of the matrix belong to the same row or column then this property is not lost after a permutation of rows and columns \cite{HM75}.

The following theorem, borrowed from \cite{HM75}, suggests a sort of matrix canonical form
that involves chainability.

\begin{lemma}\label{lem:diagblock}
If $A$ is $m\times n$ and has no rows or columns of zeros, then there are permutations matrices $P$ and $Q$ so that
$$PAQ=\begin{bmatrix}A_1 & 0      & \cdots & 0\\
                      0  & A_1    & \cdots & 0\\
                  \vdots & \vdots &        & \vdots\\
                      0  & 0      & \cdots & A_s                 \end{bmatrix}$$
where the diagonal blocks $A_k$, for $k=1,2,\ldots,s$, are chainable.
\end{lemma}

Finally, we are ready to prove the following characterization of arbitrarily regularizable graphs.

\begin{theorem}\label{th:negregchar} Let $A$ be a square $n\times n$  matrix,   and let
$G_A$ be the graph whose adjacency matrix is $A$. Then
\begin{enumerate}
\item
$G_A$ is arbitrarily regularizable if and
only if there exist two permutation matrices $P$ and $Q$ such that $PAQ$ is a block diagonal matrix with square and chainable diagonal blocks;
\item $G_A$ is not arbitrarily regularizable if and
only if there exist two permutation matrices $P$ and $Q$ such that $PAQ$ is a block diagonal matrix with chainable diagonal blocks some of which are not square.
\end{enumerate}
\end{theorem}

\begin{proof}
We start by proving item (1).
Let us assume first that there exist two permutation matrices $P$ and $Q$ such that $PAQ$ is block diagonal with square and chainable diagonal blocks.
Theorem 1.2 in \cite{HM75} states that the graph $G_A^*$ is connected if and only if $A$ is chainable. Hence, each diagonal block corresponds to a connected component of  $G_{PAQ}^*$.
Since the diagonal blocks are square the connected components of $G_{PAQ}^*$ are balanced and thus
arbitrarily regularizable. This implies that $G_{PAQ}^*$ is arbitrarily regularizable.
Hence $G_{A}^*$ is  arbitrarily regularizable, being isomorphic to $G_{PAQ}^*$. We obtain the thesis by means of Theorem \ref{th:negreg3}.

Now let us assume that $G_A$ is arbitrarily regularizable. This implies that $A$ cannot contain rows or columns of zeros, so that,
by means of Lemma \ref{lem:diagblock} we obtain that there exist two permutations $P$ and $Q$ such that $PAQ$ is block diagonal with chainable diagonal blocks. Since each of the diagonal blocks corresponds to a connected component of $G_{PAQ}^*$ the presence of non-square blocks would imply the presence of unbalanced connected components in $G_{PAQ}^*$. Hence $G_{PAQ}^*$ would be not arbitrarily regularizable. But this is impossible since $G_{PAQ}^*$ is isomorphic to $G_A^*$. Hence all the diagonal blocks must be square.

Now to prove item (2) we note that it is impossible to find two permutations $P$ and $Q$ such that $PAQ$ is block diagonal with square and chainable diagonal blocks and at the same time two permutations $R$ and $S$ such that $RAS$ is block diagonal with chainable diagonal blocks some of which are not square. Indeed, the two graphs $G^*_{PAQ}$ and $G^*_{RAS}$ would be isomorphic, since they are both isomorphic to $G^*_A$.
\end{proof}

\begin{example}
Let us consider the adjacency matrix $A$ of the top right graph in Figure \ref{fig:undirected}.
$$
A=\begin{bmatrix}     0 & 1 & 1 & 1 & 0 & 0\\
                      1 & 0 & 0 & 0 & 0 & 0\\
                      1 & 0 & 0 & 0 & 0 & 0\\
                      1 & 0 & 0 & 0 & 1 & 1\\
                      0 & 0 & 0 & 1 & 0 & 0\\
                      0 & 0 & 0 & 1 & 0 & 0
\end{bmatrix}.$$
The matrix is not chainable. Observe that the equivalence classes of the relation $\leftrightarrow$ are $E_1=\{(1,2)$, $(1,3)$, $(1,4)$, $(5,4)$, $(6,4)\}$
and $E_2=\{(2,1)$, $(3,1)$, $(4,1)$, $(4,5)$, $(4,6)\}$. By using two permutation $P$ and $Q$ to move rows $1$, $5$, $6$ (the white nodes of $E_1$) on the top of the matrix
 and rows $2$, $3$, $4$ (the white nodes of $E_2$) on the bottom and to move columns $2$, $3$, $4$ (the black nodes of $E_1$)
 on the left and columns $1$, $5$, $6$ (the black nodes of $E_2$) on the right
we obtain
$$PAQ= \begin{bmatrix} 1 & 1 & 1 & 0 & 0 & 0\\
                       0 & 0 & 1 & 0 & 0 & 0\\
                       0 & 0 & 1 & 0 & 0 & 0\\
                       0 & 0 & 0 & 1 & 0 & 0\\
                       0 & 0 & 0 & 1 & 0 & 0\\
                       0 & 0 & 0 & 1 & 1 & 1
\end{bmatrix}.$$
Observe that $PAQ$ is block diagonal with square and chainable diagonal blocks. Hence $G_A$ is arbitrarily regularizable.

As a second example, the matrix
$$
A_1=\begin{bmatrix}   0 & 1 & 0 & 0\\
                    0 & 0 & 1 & 0\\
                    1 & 0 & 0 & 1\\
                    0 & 1 & 1 & 0
\end{bmatrix}
$$
is the adjacency matrix  of the top left graph in Figure \ref{fig:directed}. The matrix is not chainable.
The equivalence $\leftrightarrow $  has the two equivalence classes  $E_1=\{(1,2)$, $(4,2)$, $(4,3)$, $(2,3)\}$ and $E_2=\{(3,1)$, $(3,4)\}$. By permuting rows and columns of $A_1$ according to the black and white nodes that appear in the two equivalence classes of the relation
$\leftrightarrow$ we obtain
$$
\begin{bmatrix}     1 & 0 & 0 & 0\\
                    0 & 1 & 0 & 0\\
                    1 & 1 & 0 & 0\\
                    0 & 0 & 1 & 1
\end{bmatrix}.
$$
The presence of non-square chainable diagonal blocks implies that $G_{A_1}$ cannot be arbitrarily regularizable.
\end{example}

\section{The regularization hierarchy is strict} \label{inclusion}

\begin{figure}[t]
\begin{center}
\includegraphics[scale=0.35, angle=0]{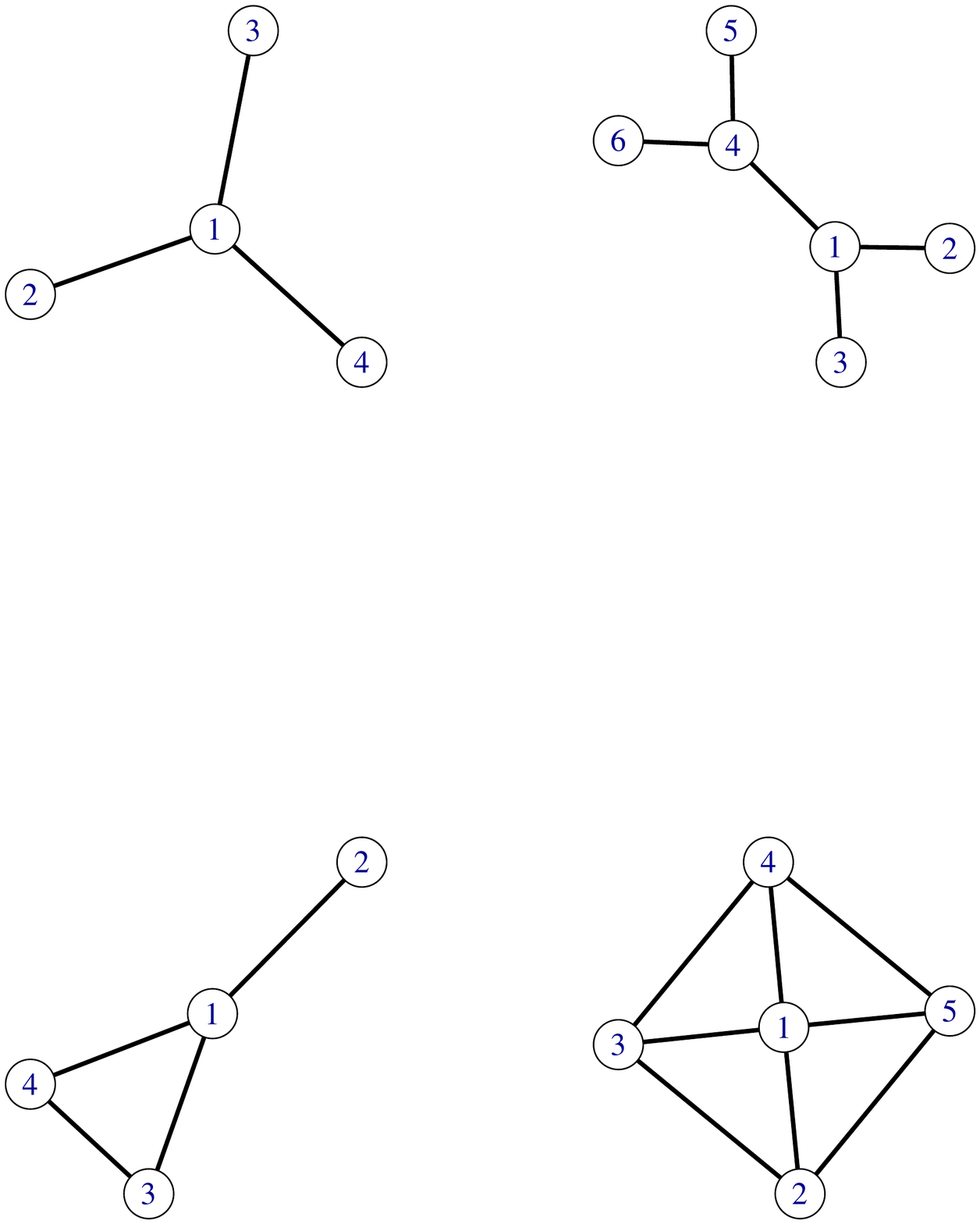}
\end{center}
\caption{A hierarchy of undirected graphs. Top-left: not arbitrarily regularizable; Top-right: arbitrarily regularizable not nonnegatively regularizable; Bottom-left: nonnegatively regularizable not positively regularizable; Bottom-right: positively regularizable not regular. }
\label{fig:undirected}
\end{figure}

\begin{figure}[t]
\begin{center}
\includegraphics[scale=0.35, angle=0]{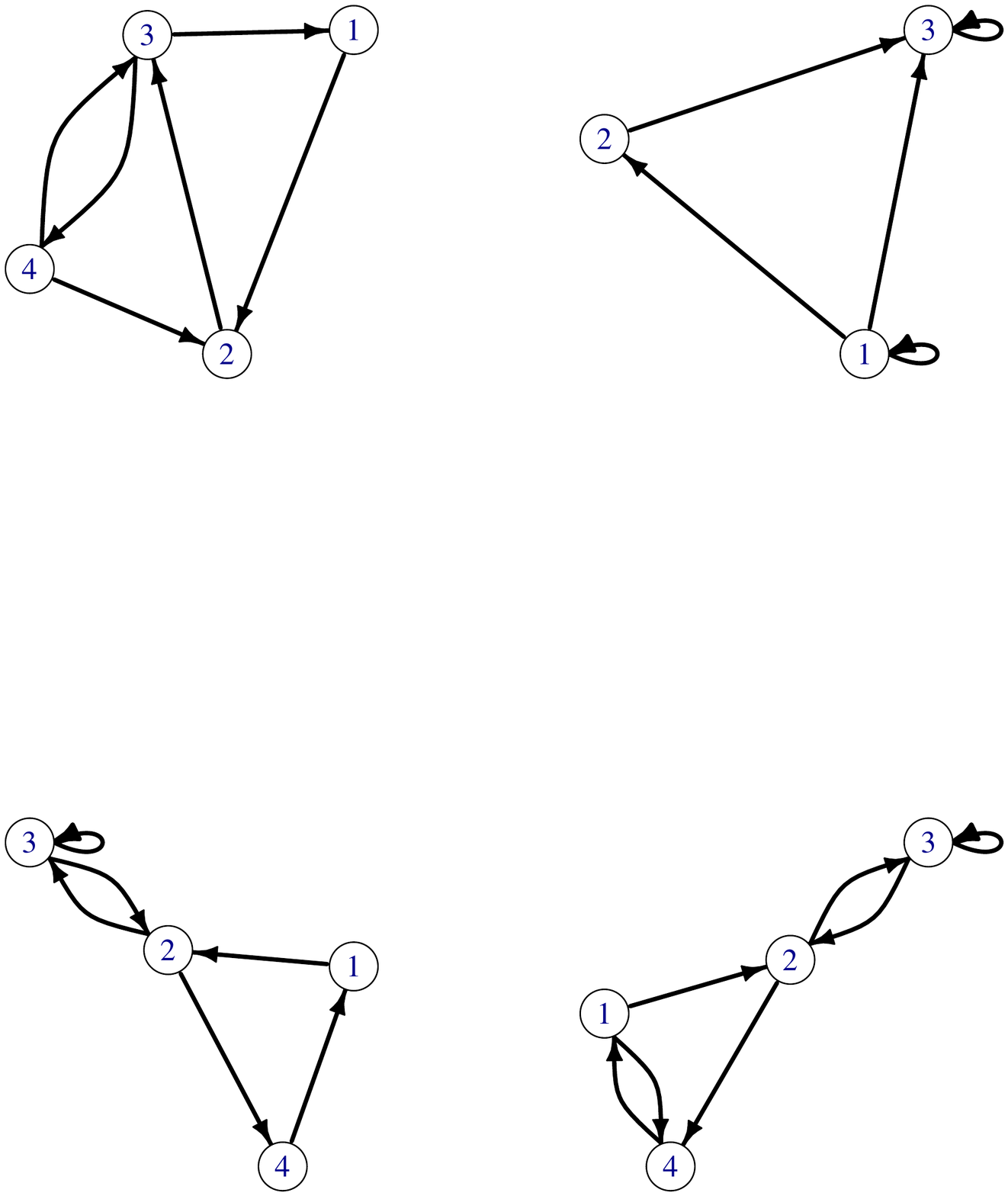}
\end{center}
\caption{A hierarchy of directed graphs. Top-left: not arbitrarily regularizable; Top-right: arbitrarily regularizable not nonnegatively regularizable; Bottom-left: nonnegatively regularizable not positively regularizable; Bottom-right: positively regularizable not regular. }
\label{fig:directed}
\end{figure}

We show here that the regularization hierarchy is strict, meaning that each class in properly contained in the previous one, for both undirected and directed graphs. We first address the undirected case. Consider Figure \ref{fig:undirected}. The top-left graph is not arbitrarily regularizable, since it is bipartite and unbalanced. In particular, since each leaf of the star must have the same degree, each edge must be labelled with the same weight $\alpha$, but this forced the degree of the center to be $3\alpha > \alpha$, unless $\alpha=0$.

A graph that is arbitrarily regularizable but not nonnegatively regularizable is the top-right one. Since the graph is bipartite and balanced, it is arbitrarily regularizable: if we label each external edge with $\alpha > 0$ and the central bridge with $-\alpha$, then each edge has degree $\alpha$. The graph is not nonnegatively regularizable since it contains no spanning cycle forest, hence it does not have support.

A graph that is nonnegatively regularizable but not positively regularizable is the bottom-left one. The graph is nonnegatively regularizable since edges $(1,2)$ and $(3,4)$ form a spanning cycle forest. The solution that labels the edges $(1,2)$ and $(3,4)$ with $\alpha > 0$ and the other edges with $0$ is a nonnegative regularizability solution with regularization degree $\alpha$. The graph is not positively regularizable since edges $(1,3)$ and $(1,4)$ do not belong to any spanning cycle forest.

Finally, a graph that is positively regularizable but not regular is the bottom-right one. The graph is positively regularizable since each edge belongs to the spanning cycle forest formed by the tringle that contains the edge plus the opposite edge. If $\alpha > 0$ and we label the outer edges with $3 \alpha$ and the inner edges with $2 \alpha$ we have a positive regularizability solution with regularization degree $8 \alpha$. Clearly, the graph is not regular.

We now address the directed case. Consider Figure \ref{fig:directed}. The top-left graph is not arbitrarily regularizable. The equivalence relation $\leftrightarrow $  has the two equivalence classes  $E_1=\{(1,2)$, $(4,2)$, $(4,3)$, $(2,3)\}$ and $E_2=\{(3,1)$, $(3,4)\}$. Class $E_1$ is unbalanced since it contains three white nodes (1, 2 and 4) and two black nodes (2 and 3). Also, class $E_2$ is unbalanced since it contains one white node (3) and two black nodes (1 and 4).

A graph that is arbitrarily regularizable but not nonnegatively regularizable is the top-right one. The graph is chainable and the equivalence relation $\leftrightarrow$ has only one class containing all edges. All nodes are both white and black and hence the class is balanced. If we weight each edge with 1 excluding edge $(1,3)$ that we weighted with -1, then all nodes have in and out degrees equal to 1. The graph is not nonnegatively regularizable since there is no spanning cycle forest: the only cycles are indeed the two loops, which do not cover node 2.

A graph that is nonnegatively regularizable but not positively regularizable is the bottom-left one. Indeed, the loop plus the 3-cycle make a spanning cycle forest. However, the remaining edges (those on the 2-cycle) are not contained in any spanning cycle forest.

Finally, a graph that is positively regularizable but not regular is the bottom-right one. The loop plus the 3-cycle and the two 2-cycles form two distinct spanning cycle forests that cover all edges. It is easy to see that the graph is not regular.

\section{Computational complexity} \label{complexity}

In this section we make some observations on the computational complexity of positioning a graph in the hierarchy we have developed. The reduction $G^*$ of a directed graph $G$ turns out to be useful to check whether a graph is nonnegatively as well as positively regularizable. We remind that a \textit{matching} $M$ is a subset of edges with the property that different edges of $M$ cannot have a common endpoint. A matching $M$ is called \text{perfect} if every node of the graph is the endpoint of (exactly) one edge of $M$. Notice that a bipartite graph has a spanning cycle forest if and only if it has a perfect matching. Moreover, every edge of a bipartite graph is included in a spanning cycle forest if and only if every edge of the graph is included in a perfect matching. Using Theorem \ref{th:regularizability}, we hence have the following.

\begin{theorem} \label{th:reg}
Let $G$ be a directed  graph and $G^*$ its undirected counterpart. Then:

\begin{enumerate}
\item $G$ is nonnegatively regularizable if and only if $G^*$ has a perfect matching;
\item $G$ is positively regularizable if and only if every edge of $G^*$ is included in a perfect matching.
\end{enumerate}

\end{theorem}

The easiest problem is to decide whether a graph is arbitrarily regularizable. For an undirected graph $G$ with $n$ nodes and $m$ edges, it involves finding the connected components of $G$ and determining if they are bipartite, and in case, if they are balanced. For a directed graph $G$, it involves finding the connected components of the undirected graph $G^*$ (which are bipartite graphs) and determining if they are balanced. All these operations can be performed in linear time $O(n+m)$ in the size of the graph. The problem can be formulated as the following linear programming feasibility problem (any feasible solution is a regularization solution):
$$
\begin{array}{l}
B w = re \\
r \geq 1
\end{array}
$$
where $B$ is the incidence matrix of the graph as defined at the beginning of Section \ref{hierarchy}.
Also regular graphs can be checked in linear time $O(n+m)$ by computing the indegrees and outdegrees for every node in directed graphs, or simply the degrees in the undirected case.
The complexity for nonnegatively and positively regularizability is higher, but still polynomial.
To determine if an undirected graph $G$ is nonnegatively regularizable, we have to find a spanning cycle forest. Using the construction adopted in the proof of Theorem 6.1.4 in \cite{LP86}, this boils down to solve a perfect matching problem on a bipartite graph of the same asymptotic size of $G$ (precisely, with a double number of nodes and edges). This costs $O(\sqrt{n} m)$ using Hopcroft-Karp algorithm for maximum cardinality matching in bipartite graphs, that is $O(n^{1.5})$ on sparse graphs. The directed case is covered by Theorem \ref{th:reg} with the same complexity. Moreover, the problem can be encoded as the following linear programming feasibility problem:
$$
\begin{array}{l}
B w = re \\
w \geq 0 \\
r \geq 1
\end{array}
$$

To decide if an undirected graph is positively regularizable, we have to find a spanning cycle forest for every edge of $G$. Using again Theorem 6.1.4 in \cite{LP86}, this amounts to solve at most $m$ perfect matching problems on bipartite graphs of the same asymptotic size of $G$, which costs $O(\sqrt{n} m^2)$, that is $O(n^{2.5})$ on sparse graphs.
Again, the directed case is addressed by Theorem \ref{th:reg} with the same complexity. Finally, the problem is equivalent to the following linear programming feasibility problem:
$$
\begin{array}{l}
B w = re \\
w \geq 1
\end{array}
$$

\section{Related literature} \label{related}

Regularizable graphs were introduced and studied by Berge \cite{B78, B81}, see also Chapter 6 in \cite{LP86}. We summarize in the following the main results related to our work. A connected undirected graph $G$ is nonnegatively regularizable (quasi-regularizable) if and only if for every independent set $S$ of nodes of $G$ it holds that $|S| \leq |N(S)|$, where $N(S)$ is the set of neighbors of $S$.
A connected undirected graph $G$ is positively regularizable (regularizable) if and only if $G$ is either elementary bipartite or 2-bicritical. A bipartite graph is elementary if and only if it is connected and each edge is included in a perfect matching. A graph $G$ is 2-bicritical if and only if for every nonempty independent set $S$ of nodes of $G$ it holds that $|S| < |N(S)|$.

In \cite{BFR15} the vulnerability of an undirected graph $G$ is defined as $$\bar{\nu}_G=\max_{S}\left( |S|-|N(S)| \right),$$ where $S$ is any independent nonempty set of nodes of $G$. It holds that $\bar{\nu}_G\leq 0$ if and only if $G$ is nonnegatively regularizable. In addition $\bar{\nu}_G< 0$ if and only if $G$ is 2-bicritical. Hence, nonnegatively regularizable graphs, and in particular, positively regularizable ones tend to have low vulnerability. On the other hand, this does not hold for arbitrarily regularizable graphs: as an example, consider the square $n\times n$ matrix

$$
A=\begin{bmatrix}0& 1 & \cdots & 1 & 1\\
             1& 0 & \cdots &0  & 1\\
        \vdots& \vdots &0  &\vdots & \vdots\\
          1& 0 & \cdots &0  & 1\\
1& 1 & \cdots & 1 & 0
\end{bmatrix}.
$$

Matrix $A$ is chainable and hence $G_A$ is arbitrarily regularizable. It is not difficult to show that $\bar{\nu}_{G_A}=n-4$, hence the vulnerability of $G_A$ can be arbitrarily high as the graph grows.

The problem or regularizability could be seen as a member of a wide family of problems concerning the existence of matrices
with prescribed conditions on the entries and on sums of certain subsets of the entries, typically row and columns \cite{HHS97,JS00}.
Additional conditions on the matrices can be imposed, such as, for example, symmetry or skew-symmetry \cite{EJKS02}.
Actually, Brualdi in \cite{BRU68} gives necessary and sufficient conditions for the existence of a nonnegative rectangular matrix with a given zero pattern and prescribed row and column sums. In \cite{HHS97} these results are generalized in various directions and in particular by considering the prescription that the entries of the matrix belong to finite or (half)infinite intervals (thus encompassing the case where some entries are prescribed to be zero or positive or nonnegative or are unrestricted).
Our approach to regularizability is more specific and graph-theoretic.

\section{Conclusion} \label{conclusion}

In many real-world social systems, links between two nodes can be represented as signed networks with positive and negative connections. With roots in social psychology, signed network analysis has attracted much attention from multiple disciplines such as physics and computer science, and has evolved both from graph-theoretic and data science perspectives \cite{TCAL16}.

In this work, we continued this endeavor by extending the notion of regularization to signed networks. We found different characterizations of the class of arbitrary regularizable graphs in terms of the topology of the graph and of the pattern of the corresponding adjacency matrix. Furthermore, we investigated the computational complexity of the problem, which can be modelled in linear programming.

\bibliographystyle{elsarticle-num}

\end{document}